\documentclass[12pt]{article}
\usepackage{amsmath,amsthm,amssymb}
\usepackage{color}
\usepackage{graphicx,subfigure,float}

\usepackage{multirow}

\usepackage[authoryear]{natbib}
\usepackage{hyperref}
\usepackage{enumitem}

\usepackage{xr}
\newcommand{\blind}{0}
\newcommand{\suit}[1]{\left(#1\right)}
\newcommand{\abs}[1]{\left\vert#1\right\vert}
\newcommand{\set}[1]{\left\{#1\right\}}
\newcommand{\msuit}[1]{\left[ #1 \right]}

\newtheorem{theorem}{Theorem}
\newtheorem{prop}{Proposition}
\newtheorem{corollary}{Corollary}

\newtheorem{remark}{{Remark}}

\begin{document}

\def\spacingset#1{\renewcommand{\baselinestretch}%
{#1}\small\normalsize} \spacingset{1}


\if0\blind
{
  \title{\bf Distributed inference for tail risks}
  \author{Liujun Chen
  \\
  International Institute of  Finance, School of Management, \\
   University of Science and Technology of China\\ 
    Deyuan Li \\
    School of Management,     Fudan University \\
    and  \\
    Chen Zhou \\
    Econometric Institute, Erasmus University Rotterdam}
  \maketitle
} \fi

\if1\blind
{
  \bigskip
  \bigskip
  \bigskip
  \begin{center}
    {\LARGE\bf Title}
\end{center}
  \medskip
} \fi

\bigskip
\begin{abstract}
For measuring tail risk with scarce extreme events, extreme value analysis is often invoked as the statistical tool to extrapolate to the tail of a distribution. The presence of large datasets benefits tail risk analysis by providing more observations for conducting extreme value analysis. However, large datasets can be stored distributedly preventing the possibility of directly analyzing them. 
 In this paper, we introduce a comprehensive set of tools for examining the asymptotic behavior of tail empirical and quantile processes in the setting where data is distributed across multiple sources, for instance, when data are stored on multiple machines. Utilizing these tools, one can establish the oracle property for most distributed estimators in   extreme value statistics in a straightforward way. The main theoretical challenge arises when the number of machines diverges to infinity. The number of machines  resembles the role of dimensionality in high dimensional statistics. 
  We provide various examples to demonstrate the practicality and value of our proposed toolkit.
\end{abstract}

\noindent%
\vfill

\newpage
\spacingset{1.8} 



\section{Introduction}
Financial  risk management requires risk forecasting for rare but high-impact events,  typically referred to as extreme events. Extreme value analysis, statistical methods for analyzing the tail of a distribution, is recognized as a useful tool for modeling and analyzing such extremes. In this paper, we consider tail risk analysis using a large dataset that is distributedly stored at various locations.

While the availability of large datasets in general benefits statistical analysis, such as extreme value analysis, it also 
presents at least three practical challenges to implementing conventional statistical  procedures. Firstly, a combined dataset might not be available to one end user due to privacy concerns.  For example, consider  analyzing insurance claims from various insurance companies. Since insurance firms are
contracted  for protecting  privacy of their customers,
 it is impossible to combine  all claims from different insurers into one massive dataset.
 Secondly,  the computation cost to analyze a massive dataset can be expensive when implementing statistical procedures involving 
an optimization algorithm, such as maximum likelihood or loss minimization. Thirdly, storage constraints can arise when dealing with massive datasets,  for instance, when the size of a dataset exceeds a computer's memory. Another example is to analyze online stream data, where data become available in a sequential manner \citep{gama2013evaluating}.


One solution to overcome these challenges is to handle the massive datasets in batches, sometimes referred to as ``distributedly stored''.
Divide and Conquer (DC) algorithms are often invoked 
 when data are distributedly stored in multiple machines. One  first  estimates a desired quantity or parameter based on part of the data  stored on each machine  and then combines the results obtained from all machines to calculate an aggregated estimate, often by a simple average.   DC algorithms have at least three advantages. Firstly, DC algorithms help to preserve privacy. For example, insurance firms can share some statistical results provided that that other companies cannot infer client level data from the
 shared results.
  Moreover, DC algorithms can significantly improve  computational efficiency by utilizing  parallel computing. Lastly, DC algorithms can  overcome the challenge of storage constraint by analyzing the dataset in batches.

  Theoretically a DC algorithm can be applied to a given statistical procedure only if it possesses the so-called  oracle property: the aggregated estimator by averaging  achieves the same statistical efficiency as the imaginary estimator using all observations. The latter  is often referred to as  the {\it oracle estimator}. The validity of applying a  DC algorithm is not obvious for many statistical procedures, and often requires additional conditions,  see \cite{fan2019distributed} for principal component analysis, \cite{volgushev2019distributed} for quantile regression and \cite{li2013statistical} for kernel density estimation, among others.

  The validity of a DC algorithm to extreme value analysis may fail.
  For example,
  considering a distribution with a finite endpoint, a natural estimator for the endpoint is the sample
  maxima. If an oracle sample maxima cannot be obtained, one may consider collecting the maxima
  from  data stored in each machine. Clearly, to obtain the oracle estimator,
  one needs to take the maximum of the machine-wise maxima instead of  taking average. Therefore, in this specific example, the
  standard DC algorithm based on averaging may lead to an estimator that fails  the oracle
  property. Since estimators in extreme value analysis are often based on high order statistics, the oracle property of the  DC algorithm for such estimators are yet to be established. The goal of this paper is to fill in this gap.

  In this paper,  we provide general tools to establish the oracle property for distributed estimators in extreme value statistics.
\cite{chen2021distributed} adapts a particular extreme value estimator,  the Hill estimator,  to the DC algorithm and proposes the distributed Hill estimator  to estimate the extreme value index $\gamma$  for the case $\gamma>0$. They study the asymptotic behavior of the  distributed Hill estimator and show
sufficient, sometimes also necessary,  conditions under which the  distributed Hill estimator possesses the oracle property.  The proof therein relies on the specific construction of the Hill estimator, and  cannot be  generalized to validate the oracle property of other  estimators for the extreme value index, let alone that of other estimators for practically relevant quantities such as high quantile, tail probability and endpoint.   By contrast, using the tools in this paper, the oracle property for most   estimators  extreme value statistics based on the peak-over-threshold (POT) approach can be established in a straightforward way.

In classical extreme value statistics, two key tools for establishing
asymptotic theories  are the tail empirical process and the tail quantile process.  
Let $l = l(N)$ be an intermediate sequence such that as $N\to\infty$, $l\to\infty, l/N\to 0$. The tail empirical process is defined as 
$$
Y_{N,l}\suit{x} = \frac{N}{l}\bar{F}_N\set{a_0\suit{\frac{N}{l}}x+b_0\suit{\frac{N}{l}}}, \quad x\in \mathbb{R},
$$
where  $\bar{F}_N:=1-F_N$ and $F_N$ denotes the empirical cumulative distribution function 
$F_N(x) := N^{-1}\sum_{i=1}^N I\suit{X_i\le x}$. Here $a_0$ and $b_0$ are suitable versions of $a$ and $b$, respectively. \cite{drees2006approximations} shows a weighted approximation of the tail empirical process $Y_{N,l}(x)$ (see Section \ref{sec:empirical} below)  under some mild conditions. The approximation of the tail empirical process is a useful tool in a wider context. For example,  \cite{drees2006approximations} proposes a test for the   extreme value condition,  \citet[Example 5.1.5]{de2006extreme} establishes the asymptotic normality of  the Hill estimator, both by using this result.

Analogous to the tail empirical process, 
\cite{drees1998smooth} shows a weighted approximation of the tail quantile process (see Section \ref{sec:quantile} below). The tail quantile process  is defined as
$$
Q_{N,l}(s) = \frac{X_{N-[ls],N} - b_0\suit{\frac{N}{l}}}{a_0\suit{\frac{N}{l}}} , \quad s\in [0,1],
$$
where  $X_{N,N}\ge \dots \ge X_{1,N}$ are the 
  order statistics of the  sample $\set{X_1,\dots,X_N}$. Here and thereafter, we use $[x]$ to denote the largest integer less than or equal to $x$. Note that the POT approach in extreme value statistics often uses high order statistics $X_{N,N}, \dots, X_{N-l,N}$. Consequently, 
compared to the tail empirical process, the tail quantile process is  more straightforward for proving asymptotic theory for estimators  in extreme value statistics based on the POT approach.  By writting such estimators  as a  functional of $Q_{N,l}(s)$ and using the weighted approximation of the tail quantile process,   one can  derive their asymptotic behavior.  Examples of  estimators for the extreme value index based on the POT approach are the probability weighted moment estimator \citep{hosking1987parameter}, the maximum likelihood estimator \citep{drees2004maximum} and the Pickands estimator \citep{pickands1975statistical}.  In addition, the asymptotic behavior for the estimators of high quantile, tail probability and endpoint can be  derived from the approximation of the tail quantile process as well, see e.g. Chapter 4 in \cite{de2006extreme}.

In this paper, we  establish weighted approximations of tail empirical processes and tail quantile processes  for  the  distributed subsamples in a joint manner,  with linking these approximations to that for the oracle sample. Such results lead to the oracle properties of many distributed estimators in extreme value statistics based on
the POT approach, with some straightforward proofs.
Mathematically, we show a stronger result than the classical oracle property:
the  difference between the distributed estimator  and the oracle estimator diminishes faster than the speed of convergence of the oracle estimator.   Note that \cite{chen2021distributed} only shows that the limiting distribution of the distributed Hill estimator coincides with that of the oracle estimator, whereas asymptotic behavior of the difference between the two estimators cannot be derived using the method therein. \cite{daouia2021optimal} achieves such result  for the   Hill estimator.
However,   the proof in \cite{daouia2021optimal} cannot   be generalized  to  other estimators in extreme value statistics since it relies on the specific structure of the Hill estimator.

The main challenge for handling  tail empirical process arises when the number of machines diverges to infinity. The number of  machines resembles the role of dimensionality in high dimensional statistics.   
Observing that with equal subsample sizes across different machines, the  tail empirical process for the oracle sample is the average of the tail empirical processes based on the distributed subsamples, it  seems trivial that they can be approximated by the same asymptotic limits. However, to aggregate the tail empirical processes in different machines, we need to make sure that
   the  approximation errors  in different machines  are uniformly negligible.   We achieve this mathematically difficult result by invoking  Koml\'os-Major-Tusn\'ady type inequalities (see e.g. \cite{komlos1975approximation}). 
    Linking the weighted approximation of the tail empirical process based on the oracle sample to those of the tail empirical processes on each machine is an important intermediate step towards establishing similar links between the corresponding tail quantile processes.

By contrast, when handling tail quantile processes, we cannot follow similar steps as  for tail empirical processes. The main difference is that the average of the tail quantile processes based on distributed subsamples in different machines is not equal to the tail quantile process based on the oracle sample. Linking the approximations of the tail quantile processes based on the distributed subsamples to that based on the oracle sample poses an additional layer of technical difficulty, which we will handle in Section  \ref{sec:quantile}.

The rest of this paper is organized as follows. Section \ref{sec:empirical} shows the weighted approximations  of the  tail empirical processes based on the distributed subsamples in a joint manner and links that to the weighted approximation of tail empirical process based on the oracle sample. Section \ref{sec:quantile} shows the analogous result for  the weighted approximations of the tail quantile processes. We provide various examples in Section \ref{Section:example} to show how these  tools can be used to prove the oracle property of  extreme value estimators such as the estimators of extreme value index, high quantile, tail probability and endpoint.  Section \ref{sec:heter}  extends the theoretical  results to the case of  heterogeneous subsample sizes.  A real data application is given in Section \ref{sec:realdata}. A concluding remark is made in Section \ref{sec:discuss}.
The technical proofs are  deferred to  the Supplementary Material, along with a simulation study showing the  performance of the distributed estimators for the extreme value index and the high quantile.

Throughout the paper, $a(t)\sim b(t)$  means that $a(t)/b(t) \to 1$ as $t\to\infty$; $a(t) \asymp b(t)$ means that both $\abs{a(t)/b(t)}$ and $\abs{b(t)/a(t)}$ are $O(1)$ as $t\to\infty$.

\section{Distributed Tail Empirical Process} \label{sec:empirical}

Let $X_1,\dots,X_N$ be independently and identically distributed (i.i.d.) random variables with distribution function $F$, which is in the maximum domain of attraction of an extreme value distribution $G_{\gamma}$ with index $\gamma\in \mathbb{R}$, i.e. there exist a positive function $a$ and a real function $b$ such that,
$$
    \lim_{N\to\infty} F^N\set{a(N) x +b(N)}= G_{\gamma}(x):=\exp\set{-\suit{1+\gamma x}^{-1/\gamma}},
$$
for all $1+\gamma x>0$.
We denote this assumption as $F\in D\suit{G_{\gamma}}$, where $\gamma$ is the so called extreme value index.
Extreme value statistics considers estimating the extreme value index $\gamma$, the functions $a$ and $b$, as well as other practically relevant quantities such as high quantile of $F$. For established results in extreme value statistics, we refer interested readers  to monographs such as \cite{de2006extreme} and \cite{resnick2007heavy}.

Write $U = \set{1/(1-F)}^{\leftarrow}$, where $^{\leftarrow}$ denotes the left-continuous inverse function. Then the necessary and sufficient condition for $F\in D(G_{\gamma})$ with $\gamma\in \mathbb{R}$ is 
\begin{equation}\label{Def: DoA condition}
    \lim_{t\to\infty} \frac{U(tx)-U(t)}{a(t)} = \frac{x^{\gamma}-1}{\gamma},
\end{equation}
for all $x>0$. In this paper, we  focus on the distributions which satisfy the second order refinement of condition \eqref{Def: DoA condition} as follows \citep{de1996generalized}: there exists an eventually positive or negative function $A$ with $\lim_{t\to\infty} A(t) = 0$ and a real number $\rho<0$ such that for all $x>0$,
\begin{equation}\label{Def:SOC:Uni}
    \lim_{t\to \infty} \frac{\frac{U(tx)-U(t)}{a(t)}-\frac{x^{\gamma}-1}{\gamma}}{A(t)}=\frac{1}{\rho}\suit{\frac{x^{\gamma+\rho}-1}{\gamma+\rho}-\frac{x^{\gamma}-1}{\gamma}}.
\end{equation}

Under this condition, one can  find suitable normalizing functions  such that the convergence in  \eqref{Def:SOC:Uni} holds uniformly as follows, see Corollary 2.3.7 in \cite{de2006extreme}.  There exists  functions $a_0(t)\sim a(t)$, $A_0(t)\sim A(t)$ and $b_0(t)$ such that, for any $\varepsilon,\delta>0$, there exists $t_0 = t_0(\varepsilon,\delta)$ such that, for all $t, tx \ge t_0$,
\begin{equation}\label{Eq:SOC:uni:uniform}
   \abs{\frac{\frac{U(tx)-b_0(t)}{a_0(t)}-\frac{x^{\gamma}-1}{\gamma}}{A_0(t)}-\Psi(x)}\le \varepsilon x^{\gamma+\rho}\max\suit{x^{\delta},x^{-\delta}},
\end{equation}
where 
$$
\Psi(x):=\left \{ 
\begin{array}{cc}
    \frac{x^{\gamma+\rho}}{\gamma+\rho}, &\gamma+\rho\ne 0,\\
    \log x, & \gamma+\rho=0.
\end{array}
\right.
$$
For the details of the expression of $a_0,b_0$ and $A_0$, see Corollary 2.3.7 in \cite{de2006extreme}.

Let $l = l(N)$ be an intermediate sequence such that as $N\to\infty$, $l\to\infty, l/N\to 0$. The tail empirical process for the oracle sample is defined as 
$$
Y_{N,l}\suit{x} = \frac{N}{l}\bar{F}_N\set{a_0\suit{\frac{N}{l}}x+b_0\suit{\frac{N}{l}}}, \quad x\in \mathbb{R},
$$
where  $\bar{F}_N:=1-F_N$ and $F_N$ denotes the empirical cumulative distribution function 
$F_N(x) := N^{-1}\sum_{i=1}^N I\suit{X_i\le x}$.

Under the second order condition \eqref{Def:SOC:Uni} and $\sqrt{l}A\suit{N/l} =O(1)$ as $N\to\infty$, \cite{drees2006approximations} shows that, under proper Skorokhod construction,  there exists a sequence of Brownian motions $\set{W_N^*}_{N\ge 1}$, such that, for any $v>0$, 
\begin{equation}\label{Eq: Oracle  TEP}
    \begin{aligned}
        \sup_{x\in \mathbb{D}} \set{z(x)}^{v-1/2}&\Big|\sqrt{l}\set{Y_{N,l}(x)-z(x)} -W_N^*\set{z(x)}\\
        &-\sqrt{l}A_0(N/l)\set{z(x)}^{1+\gamma}\Psi\set{1/z(x)}\Big| =o_P(1),
      \end{aligned} 
\end{equation}
where $z(x)$ and $\mathbb{D}$ are defined in Theorem \ref{Theorem: EP Divide and Conquer} below.

The result \eqref{Eq: Oracle  TEP} is based on the oracle sample.   We intend to provide an analogous result for the tail empirical processes  based on the distributed subsamples in a joint manner. Assume that  the $N$ observations are  distributedly stored in  $m$ machines  with $n$ observations in each machine and then $N = nm$. We will extend our analysis to the case of heterogeneous subsample size in Section \ref{sec:heter}.
The tail empirical process based on the observations $\set{X_1^{(j)},\dots, X_n^{(j)}  }$ in  machine  $j$ is defined as 
$$
Y_{n,k}^{(j)}(x) = \frac{n}{k} \bar{F}_n^{(j)}\set{a_0\suit{\frac{n}{k}}x+b_0\suit{\frac{n}{k}}},\quad \quad j=1,\dots,m,
$$
where $\bar{F}_{n}^{(j)}:=1-F_{n}^{(j)}$ and $F_n^{(j)}$ denotes the empirical distribution function based on the observations in machine $j$.  Here $k = k(N)$ is  an intermediate sequence such that $k\to\infty$ and $ k/n \to 0$, as $N\to\infty$. 

We intend  to relate the asymptotics of  $Y_{N,l}(x)$ and $m^{-1} \sum_{j=1}^m Y_{n,k}^{(j)}(x)$ where $l = km$. Without causing any ambiguity, we use the simplified notation $Y_{N}(x)$ and $Y_n^{(j)}(x)$ for the tail empirical processes based on the oracle sample and  the sample in machine $j$, respectively.

Throughout this paper, let $m,n,k$ be sequences of integers such that, $m = m(N)\to\infty, n =n(N)\to\infty, k=k(N)\to\infty$ and $ k/n \to 0$ as $N\to\infty$. We  assume the following conditions for the sequences $k$ and $m$:
 
\begin{enumerate}[label=(A\arabic*)]
  \item \label{Condition:u:k1}$\sqrt{km}A(n/k)=O(1)$  as $N \to \infty$.
\end{enumerate}

\begin{enumerate}[label=(A\arabic*)]
  \setcounter{enumi}{1}
  \item \label{Condition:u:k2} 
  $\eta := \liminf_{N\to\infty} \log k/\log m -1>0$.
\end{enumerate}

\begin{enumerate}[label=(A\arabic*)]
  \setcounter{enumi}{2}
  \item \label{Condition:u:k3} 
  $km(\log k)^2 /n =O(1)$ as $N\to\infty$.
\end{enumerate}

\begin{remark}
Note that $n/k = N/(nm)$. 
    Condition \ref{Condition:u:k1} is a typical condition assumed in extreme value analysis to guarantee  finite asymptotic bias in the oracle estimator. Condition \ref{Condition:u:k2} states that, the number of machines ($m$) should be smaller than the number of observations  used in each machine ($k$). Similar conditions are  assumed in the literature of distributed inference for other statistical procedures, see e.g.  Corollary 3.4 in   \cite{volgushev2019distributed} and  Theorem 4 in  \cite{zhu2021least}.
    Condition \ref{Condition:u:k3} is an additional technical condition, which requires that  the number of observations ($n$) in each machine  is  at a sufficiently high level for  given $k$ and $m$.
\end{remark}

\begin{remark}
  One example for $k$ and $m$ satisfying conditions \ref{Condition:u:k1}-\ref{Condition:u:k3}  can be given as follows. 
Let $m \asymp n^a$ for some $0\le a<\frac{(-1)\lor \rho}{1-(-1)\lor \rho}$, where $\rho$ is the second parameter in \eqref{Def:SOC:Uni}, and  $k \asymp n^b$ for some $$a<b <\min\suit{1-a,\frac{-2\rho-a}{-2\rho+1}},
$$ then  conditions \ref{Condition:u:k1}-\ref{Condition:u:k3} hold with $\eta = b/a - 1>0$.
\end{remark}

The following theorem shows the weighted approximations of the tail empirical processes based on the distributed subsamples in a joint manner.

\begin{theorem}\label{Theorem: EP Divide and Conquer}
    Suppose that the distribution function $F$ satisfies the second order condition \eqref{Def:SOC:Uni} with $\gamma \in \mathbb{R}$ and   $\rho<0$. 
    Let $m,n,k$ be sequences of integers  satisfying  
      conditions \ref{Condition:u:k1}-\ref{Condition:u:k3} and $x_0>-1/(\gamma \lor 0)$. Then under suitable Skorokhod construction,  there exist $m$ independent sequences of Brownian motions $\set{W_n^{(j)}}_{n\ge 1}, j=1,\dots,m$, such that for any  $v\in((2+\eta)^{-1},2^{-1})$, as $N \to \infty$, 
      $$
      \begin{aligned}
        \max_{1\le j\le m} \sup_{x\in \mathbb{D}} \set{z(x)}^{v-1/2} &\Big|\sqrt{km}\set{Y_n^{(j)}(x)-z(x)}-\sqrt{m}W_n^{(j)}\set{z(x)} \\
        &-\sqrt{km}A_0(n/k)\set{z(x)}^{1+\gamma}\Psi\set{1/z(x)}\Big|  = o_P(1),
      \end{aligned}
      $$
      where
      $$
      z(x)=(1+\gamma x)^{-1/\gamma}, \quad  \mathbb{D}=\set{x:x_0\le x< \frac{1}{(-\gamma) \lor 0}}.
      $$
      Moreover,  as $N \to \infty$,
      $$
      \begin{aligned}
        \sup_{x\in \mathbb{D}} \set{z(x)}^{v-1/2}\Big|&\sqrt{km}\set{Y_{N}(x)-z(x)} -W_N\set{z(x)} \\
        &-\sqrt{km}A_0(n/k)\set{z(x)}^{1+\gamma}\Psi\set{1/z(x)}\Big| =o_P(1),
      \end{aligned} 
      $$
      where  $W_N = m^{-1/2}\sum_{j=1}^m W_n^{(j)}$ is  a version of the Brownian motion $W_N^*$ in \eqref{Eq: Oracle TEP}.
\end{theorem}


For $\gamma>0$, a similar but simpler result is given as follows.
\begin{theorem}\label{Theorem: DC TEP gamma>0}
    Suppose that the distribution function $F$ satisfies the second order condition \eqref{Def:SOC:Uni} with $\gamma>0$ and   $\rho<0$. 
    Let $m,n,k$ be sequences of real numbers that satisfy 
      conditions \ref{Condition:u:k1}-\ref{Condition:u:k3} and $\tilde{x}_0> 0$.  Then under suitable Skorokhod construction, there exist $m$ independent sequences of Brownian motions $\set{W_n^{(j)}, n\ge 1}, j=1,\dots,m$, such that for any  $v\in((2+\eta)^{-1},2^{-1})$, as $N \to \infty$,  
      $$
      \begin{aligned}
      \max_{1\le j\le m}\sup_{x\ge \tilde{x}_0} x^{(1/2-v)/\gamma}\Big| &\sqrt{km}\set{\frac{n}{k}\bar{F}_n^{(j)}\suit{xU(n/k)}-x^{-1/\gamma}} \\
        & -\sqrt{m}W_n^{(j)}(x^{-1/\gamma})-\sqrt{km}A_0\suit{\frac{n}{k}} x^{-1/\gamma}\frac{x^{\rho/\gamma}-1}{\gamma\rho}
        \Big|= o_P(1).      
      \end{aligned}
      $$
      Moreover,   as $N\to\infty$,
      $$
      \begin{aligned}
       \sup_{x\ge \tilde{x}_0} x^{(1/2-v)/\gamma}\Big| &\sqrt{km}\set{\frac{n}{k}\bar{F}_N\suit{xU(n/k)}-x^{-1/\gamma}} \\
        & -W_N(x^{-1/\gamma})-\sqrt{km}A_0\suit{\frac{n}{k}} x^{-1/\gamma}\frac{x^{\rho/\gamma}-1}{\gamma\rho}
        \Big|= o_P(1),
      \end{aligned}
        $$
        where  $W_N = m^{-1/2}\sum_{j=1}^m W_n^{(j)}$.
\end{theorem}

To prove these theorems, we need a fundamental inequality to bound the approximation error of the tail empirical process $Y_{n}^{(j)}(x)$ to the Gaussian process in  machine $j$, which is of independent interest.
Consider a positive sequence $t = t(N) \to 0$ as $N\to\infty$,   satisfying 
\begin{eqnarray}
     (n/k)^{-1/2}\log k/t&=&O(1),\label{Eq: r 1}\\
    k^{1/2}A_0(n/k)/t& = &O(1), \label{Eq: r 2}\\
    \text{for some}\   \tilde{\varepsilon}>0,  \set{A_0(n/k)}^{1/2-\tilde{\varepsilon}}/t&=&o(1). \label{Eq: r 3}  
\end{eqnarray}

\begin{prop}\label{Theorem: EP single machine}
    Suppose that the distribution function $F$ satisfies the second order condition \eqref{Def:SOC:Uni} with $\gamma \in \mathbb{R}$ and   $\rho<0$. 
    Let $t$ be a  sequence of real numbers  satisfying conditions \eqref{Eq: r 1}-\eqref{Eq: r 3} and 
    $x_0>-1/(\gamma \lor 0 )$.
    Then  for sufficiently large $n$, under suitable Skorokhod construction, there exist $m$ independent sequences of Brownian motions $\set{W_n^{(j)}, n\ge 1}, j=1,\dots,m$ and 
     a constant $C_1=C_1(v)>0$  such that, for any $v\in (0,1/2)$,
     $$    
     P\suit{ \delta_{n}^{(j)} \ge  t}\le C_1 r^{-\frac{1}{1/2-v}},
    $$
     where 
  $$
  \begin{aligned}
    \delta_{n}^{(j)} =\sup_{x\in \mathbb{D}} \set{z(x)}^{v-1/2}\Big|&\sqrt{k}\set{Y_n^{(j)}(x)-z(x)} \\
    &-W_n^{(j)}\set{z(x)}-\sqrt{k}A_0(n/k)\set{z(x)}^{1+\gamma}\Psi\set{1/z(x)}\Big|,   
  \end{aligned}
  $$
  and  $r = r(t,k)$ is defined by  $k^{-v} r\log r = t$.
  \end{prop}

  Proposition \ref{Theorem: EP single machine} guarantees that the approximation errors $\delta_{n}^{(j)}, j=1,\dots,m$ are uniformly negligible, which is a key step to prove   Theorems \ref{Theorem: EP Divide and Conquer} and \ref{Theorem: DC TEP gamma>0}.

\section{Distributed Tail Quantile Processes} \label{sec:quantile}

\cite{drees1998smooth} provides a weighted approximation of the tail quantile process. Assume the second order condition \eqref{Def:SOC:Uni} and $\sqrt{l}A\suit{N/l} =O(1)$ as $N\to\infty$, with the same   Brownian motions $\set{W_N^{*}}_{ N\ge 1}$ in \eqref{Eq: Oracle  TEP}, we have that, for any $v>0$, 
$$
    \begin{aligned}
      \sup_{1/l\le s\le 1}s^{v+1/2+\gamma} & \Big |\sqrt{l}\suit{Q_{N,l}(s)-\frac{s^{-\gamma}-1}{\gamma}}\\
        &-s^{-\gamma-1}W_N^{*}(s)-\sqrt{l}A_0\suit{\frac{N}{l}}\Psi(s^{-1})\Big |=o_P(1).
      \end{aligned} 
$$  
Again, this result is based on the oracle sample. We intend to provide  weighted approximations of the tail quantile processes based on the distributed subsamples in a joint manner.
The tail quantile process based on the observations in machine  $j$ is defined as 
$$
Q_{n,k}^{(j)}(s) = \frac{X^{(j)}_{n-[ks],n} - b_0\suit{\frac{n}{k}}}{a_0\suit{\frac{n}{k}}} , \quad \quad j=1,\dots,m,
$$
where $X^{(j)}_{n,n} \ge \cdots\ge X^{(j)}_{1,n}$ are the order statistics of the observations in machine $j$.

We aim at linking the asymptotics of $Q_{N,l}(s)$ and $m^{-1}\sum_{j=1}^m Q_{n,k}^{(j)}(s)$ where $l = km$.  Again, without causing any ambiguity, we use the simplified notation $Q_{N}(s)$ and $Q_{n}^{(j)}(s)$ for the tail quantile process based on the oracle sample and the sample in machine $j$, respectively.
Since  the average of the tail quantile processes based on distributed subsamples in $m$ machines is {\it not} equal to the tail quantile process of the oracle sample, we cannot follow similar steps as in Section \ref{sec:empirical}. Instead, we achieve our goal by  ``inverting" the result for the tail empirical processes. More specifically, we intend to  replace $x$  in Theorem \ref{Theorem: EP Divide and Conquer} by $Q_n^{(j)}(s)$ for $s\in [k^{-1+\delta}, 1]$.

The following theorem shows that,
with the same  sequences of Brownian motions defined in Theorem \ref{Theorem: EP Divide and Conquer}: $\set{W_n^{(j)}}_{n\ge 1}$, $j=1,\dots,m$, the approximation errors of the tail quantile processes
 are uniformly negligible for $1\le j\le m$.

\begin{theorem}\label{Theorem: TQP DC}
  Assume the same conditions as in Theorem \ref{Theorem: EP Divide and Conquer}. Then  for any $v\in ((2+\eta)^{-1},1/2)$ and $\delta \in (0,1)$,  as $N \to \infty$,
    $$
    \begin{aligned}
      \max_{1\le j\le m}\sup_{k^{-1+\delta}\le s\le 1} s^{v+1/2+\gamma}\Big |&\sqrt{km}\suit{Q_{n}^{(j)}(s)-\frac{s^{-\gamma}-1}{\gamma}}\\
        &-\sqrt{m}s^{-\gamma-1}W_n^{(j)}(s)-\sqrt{km}A_0\suit{\frac{n}{k}}\Psi(s^{-1})\Big |=o_P(1).  
    \end{aligned}
    $$
    Moreover,  as $N\to\infty$,
    $$
        \begin{aligned}
           \sup_{k^{-1+\delta}\le s\le 1} s^{v+1/2+\gamma}\Big |&\sqrt{km}\suit{Q_{N}(s)-\frac{s^{-\gamma}-1}{\gamma}}\\
            &-s^{-\gamma-1}W_N(s)-\sqrt{km}A_0\suit{\frac{n}{k}}\Psi(s^{-1})\Big |=o_P(1).
          \end{aligned} 
    $$
    Here, $\set{W_n^{(j)}}_{n\ge 1}, j=1,\dots,m$ are the same Brownian motions constructed as in  Theorem \ref{Theorem: EP Divide and Conquer} and  $W_N = m^{-1/2}\sum_{j=1}^m W_n^{(j)}$. 
    Consequently,  $m^{-1}\sum_{j=1}^m Q_n^{(j)}(s)$ has the same asymptotic expansion as that for $Q_N(s)$, uniformly for  $s\in [k^{-1+\delta},1]$.
\end{theorem}

For $\gamma>0$, a similar but simpler result is given as follows.
\begin{theorem}\label{Theorem: TQP DC gamma>0}
  Assume the same conditions as in Theorem \ref{Theorem: DC TEP gamma>0}. Then  for any $v\in ((2+\eta)^{-1},1/2)$ and $\delta \in (0,1)$,  as $N \to \infty$,
    $$
    \begin{aligned}
        \max_{1\le j\le m}\sup_{k^{-1+\delta}\le s\le 1} &s^{v+1/2+\gamma}\Big |\sqrt{km}\suit{\frac{X_{n-[ks],n}^{(j)}}{U(n/k)}-s^{-\gamma}}\\
        &-\sqrt{m}\gamma s^{-\gamma-1}W_n^{(j)}(s)-\gamma\sqrt{km}A_0\suit{\frac{n}{k}}s^{-\gamma}\frac{s^{-\rho}-1}{\rho}\Big |=o_P(1).  
    \end{aligned}
    $$
    Moreover, as $N\to\infty$,
    $$
    \begin{aligned}
       \sup_{k^{-1+\delta}\le s\le 1} s^{v+1/2+\gamma}\Big |&\sqrt{km}\suit{\frac{X_{N-[kms],N}}{U(n/k)}-s^{-\gamma}}\\
        &-\gamma s^{-\gamma-1}W_N(s)-\gamma\sqrt{km}A_0\suit{\frac{n}{k}}s^{-\gamma}\frac{s^{-\rho}-1}{\rho}\Big |=o_P(1).
    \end{aligned}
    $$
    Here, $\set{W_n^{(j)}, n\ge 1}, j=1,\dots,m$ are the same Brownian motions constructed as in Theorem \ref{Theorem: EP Divide and Conquer} and  $W_N = m^{-1/2}\sum_{j=1}^m W_n^{(j)}$. 
\end{theorem}

 The following corollary, which is a direct consequence of  Theorem \ref{Theorem: TQP DC gamma>0} with applying the  Cram\'er's delta method,  can be used  for proving  asymptotic theory of the distributed Hill estimator.
\begin{corollary}
  Assume  the same conditions  as in Theorem \ref{Theorem: DC TEP gamma>0}. By the Cram\'er's delta method, we can obtain that, as $N\to\infty$,
  $$
  \begin{aligned}
       \max_{1\le j\le m}\sup_{k^{-1+\delta}\le s\le 1} s^{v+1/2}\Big |&\sqrt{km}\suit{\frac{\log X_{n-[ks],n}^{(j)}-\log U\suit{\frac{n}{k}}}{\gamma}+\log s}\\
       &-\sqrt{m}\gamma s^{-1}W_n^{(j)}(s)-\gamma\sqrt{km}A_0\suit{\frac{n}{k}}\frac{1}{\gamma}\frac{s^{-\rho}-1}{\rho}\Big |=o_P(1).  
   \end{aligned}
  $$
\end{corollary}

Theorem \ref{Theorem: TQP DC}  provides a useful tool for establishing the oracle property of some extreme value estimators based on the POT approach.  For example, using Theorem \ref{Theorem: TQP DC}, one can  immediately show that, the distributed Pickands estimator achieves the oracle property  since the distributed Pickands estimator is a functional of the tail quantile processes $Q_n^{(j)}(s)$ at three points $s=1,1/2$ and $1/4$. We leave this to the readers. Instead, we focus on 
 some other  estimators, for which  Theorem \ref{Theorem: TQP DC} alone may not be sufficient for proving their oracle property. We use the probability weighted moment (PWM) estimator as an example to explain the remaining issue.

The PWM  estimator in machine $j$ is defined as 
$$
\widehat{\gamma}_{PWM}^{(j)}:=  \frac{P_n^{(j)}-4Q_n^{(j)}}{P_n^{(j)}-2Q_n^{(j)}},
$$
where
$$
P_n^{(j)}:=\frac{1}{k}\sum_{i=1}^{k}X_{n-i+1,n}^{(j)}-X_{n-k,n}^{(j)}, \quad Q_n^{(j)}:=\frac{1}{k}\sum_{i=1}^{k}\frac{i-1}{k}\suit{X_{n-i+1,n}^{(j)}-X_{n-k,n}^{(j)}}.
$$
The distributed PWM estimator is defined as the average of the $m$ estimates from each machine:
$$
\widehat{\gamma}_{PWM}^{D}  = \frac{1}{m}\sum_{j=1}^m\widehat{\gamma}_{PWM}^{(j)}.
$$
To establish the asymptotic theory for $\widehat{\gamma}_{PWM}^{D}$, we need to handle the asymptotic expansion of $P_n^{(j)}$ and $Q_n^{(j)}$ for $j=1,\dots,m$ in a joint manner.
For $s \in [0,1]$, define 
\begin{equation}\label{Def:  f_n,j,s}
f_n^{(j)}(s) = Q_{n}^{(j)}(s)-\frac{s^{-\gamma}-1}{\gamma} - \frac{1}{\sqrt{k}}s^{-\gamma-1}W_n^{(j)}(s)-A_0\suit{\frac{n}{k}}\Psi(s^{-1}).
\end{equation}
Then we can write  $P_n^{(j)}$ as
  $$
    \begin{aligned}
    &\frac{P_n^{(j)}}{a_0\suit{n/k}} \\
        =& \int_{0}^1 \frac{X_{n-[ks],n}^{(j)}-X_{n-k,n}^{(j)}}{a_0\suit{\frac{n}{k}}} ds \\
        =&\int_{0}^{1} \frac{s^{-\gamma}-1}{\gamma}ds + \frac{\int_{0}^{1} \set{s^{-\gamma-1}W_n^{(j)}(s)-W_n^{(j)}(1)}ds }{\sqrt{k}}   
        + A_0\suit{\frac{n}{k}}\int_{0}^{1}\set{\Psi(s^{-1})-\Psi(1)}  ds \\
         &+\int_{k^{-1+\delta}}^{1}\set{f_n^{(j)}(s)-f_n^{(j)}(1)} ds +  \int_{0}^{k^{-1+\delta}} \set{f_n^{(j)}(s)-f_n^{(j)}(1)}ds\\
        =&: I_1+I_2+I_3+I_4+I_5.
    \end{aligned}    
$$

The three terms $I_1, I_2$ and $I_3$   can be handled in a similar way as  handling  analogous terms in  the oracle PWM estimator.
The  integral $I_4$  can be handled using Theorem \ref{Theorem: TQP DC}. However, handling the last integral $I_5$  requires some additional tools to deal with the ``corner'' of the tail quantile processes. Similarly, for $Q_n^{(j)}$, we need to handle a different integral in the ``corner'': 
$\int_{0}^{k^{-1+\delta}} s\set{f_n^{(j)}(s)-f_n^{(j)}(1)} ds.
$
To complete the toolkit for our purpose,  we provide a general result regrading the joint asymptotic behavior of  weighted integrals of the tail quantile processes in the corner area $[0,k^{-1+\delta}]$.

\begin{prop}\label{Theorem: Corner}
    Assume  the same conditions  as in Theorem \ref{Theorem: EP Divide and Conquer}. Assume that  a function $g$  defined on (0,1) satisfies $0<g(s)\le C s^{\beta}$ with $\beta>\gamma-\frac{\eta}{2(1+\eta)}+\frac{1}{1+\eta}\gamma I\set{\gamma>0}$. Then, there exists a sufficiently small constant $\delta>0$, such that, as $N \to \infty$,
    $$
    \begin{aligned}
        \sqrt{m}\max_{1\le j \le m}\int_{0}^{k^{-1+\delta}} g(s) \Big| & \sqrt{k}\suit{Q_{n}^{(j)}(s) - \frac{s^{-\gamma}-1}{\gamma} }\\
        &-s^{-\gamma-1}W_n^{(j)}(s)-\sqrt{k}A_0\suit{\frac{n}{k}}\Psi(s^{-1})\Big| ds =o_P(1).
    \end{aligned}
        $$
\end{prop}

The oracle property of most extreme value estimators, including the PWM estimator,  in the extreme value statistics based on the POT approach  can be established by applying Theorem \ref{Theorem: TQP DC} and Proposition \ref{Theorem: Corner} together.
We demonstrate a few examples in Section \ref{Section:example}.

\section{Application}\label{Section:example}

\subsection{Distributed inference for the Hill estimator}
In this subsection, we apply the approximations of the tail empirical processes based on the distributed subsamples to establish the oracle property of the distributed Hill estimator.
The Hill estimator in machine $j$ is defined as 
$$
\widehat{\gamma}_H^{(j)}: = \frac{1}{k}\sum_{i=1}^{k} \log X_{n-i+1,n}^{(j)} - \log X_{n-k,n}^{(j)}, \quad j=1,\dots,m.
$$
The distributed Hill estimator is defined as the average of the $m$ estimates from each machine:
$\widehat{\gamma}_H^{D}:=\frac{1}{m} \sum_{j=1}^m\widehat{\gamma}_H^{(j)}$.
And the oracle Hill estimator using the top $l = km$ exceedance ratios is 
$$
\widehat{\gamma}_H^{Oracle}:=\frac{1}{km}\sum_{i=1}^{km} \log X_{N-i+1,N} - \log X_{N-km,N}.
$$
\begin{corollary}\label{Corrolary: Hill}
    Suppose that the distribution function $F$ satisfies the second order condition \eqref{Def:SOC:Uni} with $\gamma>0$ and   $\rho<0$. 
    Let $m,n,k$ be sequences of real numbers that satisfy 
      conditions \ref{Condition:u:k1}-\ref{Condition:u:k3}. Then,  the distributed Hill estimator achieves the oracle property, i.e. $\sqrt{km}\suit{\widehat{\gamma}_{H}^{D} - \widehat{\gamma}_H^{Oracle}} = o_P(1)$, as $N\to\infty$. 
\end{corollary}
\begin{proof}[Proof of Corollary \ref{Corrolary: Hill}]
  By applying the same techniques used in proving the asymptotic normality of the oracle Hill estimator (cf. Example 5.1.5 in \cite{de2006extreme}), we have that, as $N\to\infty$,
  $$
  \begin{aligned}
   \widehat{\gamma}_{H}^{D} -\gamma 
   =&\frac{1}{m}\sum_{j=1}^m \int_{X_{n-k,n}^{(j)}/U(n/k)}^{1} s^{-1/\gamma} \frac{ds}{s} \\
  &\quad + \frac{1}{m}\sum_{j=1}^m \int_{X_{n-k,n}^{(j)}/U(n/k)}^{1} \msuit{\frac{n}{k}\set{1-F_n^{(j)}\suit{sU\suit{\frac{n}{k}}}}-s^{-1/\gamma}}\frac{ds}{s} \\
  &\quad +\frac{1}{m}\sum_{j=1}^m \int_{1}^{\infty} \msuit{\frac{n}{k}\set{1-F_n^{(j)}\suit{sU\suit{\frac{n}{k}}}}-s^{-1/\gamma}}\frac{ds}{s}\\
  =&:I_1+I_2+I_3.
  \end{aligned}
  $$
  
  For  $I_1$, note that, as $N\to\infty$,   
  $$
  \begin{aligned}
    I_1 &= \frac{1}{m}\sum_{j=1}^m \set{\gamma\suit{X_{n-k,n}^{(j)}/U(n/k)}^{-1/\gamma} -\gamma}. \\
  \end{aligned}
  $$
  By taking $s=1$ in Theorem 4, we get that, as $N\to\infty$,
  \begin{equation}\label{Eq: uniform for X_{n-k,n}}
    \sqrt{m}\max_{1\le j \le m}\abs{\sqrt{k}\suit{\frac{X_{n-k,n}^{(j)}}{U(n/k)}-1}-\gamma W_n^{(j)}(1)}=o_P(1).
  \end{equation}
  Thus, as $N\to\infty$,
  $$
  I_1 = -\gamma (km)^{-1/2} \frac{1}{\sqrt{m}}\sum_{j=1}^m W_n^{(j)}(1) +(km)^{-1/2}o_P(1).
  $$
  
  For $I_2$, the uniform convergence in \eqref{Eq: uniform for X_{n-k,n}} and  Theorem \ref{Theorem: EP Divide and Conquer} imply that as $N\to\infty$, $I_2=(km)^{-1/2}o_P(1)$.
  
  For $I_3$, since $F_N=m^{-1} \sum_{j=1}^m F_n^{(j)}$, we obtain that,
  $$
  I_3 =  \int_{1}^{\infty} \msuit{\frac{n}{k}\set{1-F_N\suit{sU\suit{\frac{n}{k}}}}-s^{-1/\gamma}}\frac{ds}{s}.
  $$
  
  We can handle $\widehat{\gamma}_{H}^{Oracle}$ in a similar way and get that, 
  $$
  \widehat{\gamma}_H^{Oracle} -\gamma= -\gamma\frac{1}{\sqrt{km}} W_N(1)+(km)^{-1/2}o_P(1)+I_3.
  $$
  The Corollary is proved by noting that  $W_N = m^{-1/2}\sum_{j=1}^m W_n^{(j)}$.  
\end{proof}

\begin{remark}
  \cite{chen2021distributed} only shows that the limiting distribution of the distributed Hill estimator coincides with that of the oracle Hill  estimator, but  does not investigate  the difference between the two estimators.
\end{remark}

\subsection{Distributed inference for the PWM estimator}
In this subsection,  we take the distributed PWM estimator as an example to show how to apply   Theorem \ref{Theorem: TQP DC} and  Proposition \ref{Theorem: Corner} to establish its oracle property.
The oracle PWM estimator is defined as 
$$
\widehat{\gamma}_{PWM}^{Oracle}: = \frac{P_N -4Q_N}{P_N-2Q_N},
$$
where $P_N$ and $Q_N$ are  counterparts of $P_n^{(j)}$ and $Q_n^{(j)}$ based on the oracle sample, respectively.

\begin{corollary}\label{Corollary: PWM}
Suppose that the distribution function $F$ satisfies the second order condition \eqref{Def:SOC:Uni} with $\gamma<1/2$ and $\rho<0$. Assume that conditions \ref{Condition:u:k1}-\ref{Condition:u:k3}  hold with  $\eta > \max\set{0,\frac{2\gamma }{1/2-\gamma}}$. Then, the distributed PWM estimator achieves the oracle property, i.e., $\sqrt{km}\suit{\widehat{\gamma}_{PWM}^{D}-\widehat{\gamma}_{PWM}^{Oracle}} = o_P(1)$ as $N\to\infty$.
\end{corollary}

\begin{proof}[Proof of Corollary \ref{Corollary: PWM}]
For a continuous function $f: [0,1]\to \mathbb{R}$, define an operator  
$$
L(f) = (1-\gamma)(2-\gamma) \int_{0}^1 \set{(1-4s)-\gamma(1-2s)} \set{f(s)-f(1)}ds.
$$
It is obvious that $L$ is a linear operator.

Note that, for  the oracle PWM estimator  using top $km$
exceedances, we have that, as $N\to\infty$,
$$
\begin{aligned}
    \widehat{\gamma}_{PWM}^{Oracle} -\gamma &= \frac{1}{\sqrt{km}} L\suit{s^{-\gamma-1}W_N(s)} +A_0\suit{\frac{n}{k}}L\suit{\Psi(s^{-1})} +\frac{1}{\sqrt{km}}o_P(1),
\end{aligned}
$$
see e.g. Section 3.6.1 in \cite{de2006extreme}. 
By using  similar techniques, we obtain that, as $N\to\infty$, 
$$
\begin{aligned}
&\frac{1}{m}\sum_{j=1}^m \widehat{\gamma}_{PMW}^{(j)} - \gamma \\
&= \frac{1}{m}\sum_{j=1}^m \frac{1}{\sqrt{k}} L\suit{s^{-\gamma-1}W_n^{(j)}(s)}+A_0\suit{\frac{n}{k}}L\suit{\Psi(s^{-1})} 
 \\
&\quad + O_P(1) \max_{1\le j\le m}\int_{0}^1 \abs{f_n^{(j)}(s)} ds + O_P(1) \max_{1\le j\le m}\int_{0}^1 s\abs{f_n^{(j)}(s)} ds.\\
\end{aligned}
$$

Recall that $L$ is a linear operator and $W_N = m^{-1/2}\sum_{j=1}^m W_n^{(j)}$, we get that
$$
\frac{1}{\sqrt{km}} L\suit{s^{-\gamma-1} W_N(s)} =\frac{1}{m} \sum_{j=1}^m \frac{1}{\sqrt{k}} L\suit{s^{-\gamma-1}W_n^{(j)}(s)}.
$$
The Corollary is proved provided that,  as $N\to\infty$, $I_1:= \max_{1\le j\le m}\int_{0}^1 \abs{f_n^{(j)}(s)} ds= (km)^{-1/2}o_P(1)$ and $I_2:=\max_{1\le j\le m}\int_{0}^1 s\abs{f_n^{(j)}(s)} ds =(km)^{-1/2} o_P(1)$.

For handling $I_1$, we  divide $[0,1]$ into
 $[k^{-1+\delta},1]$ and $[0,k^{-1+\delta}]$.  Thus,
 $$
 \begin{aligned}
    I_1 &\le \max_{1\le j\le m}\int_{0}^{k^{-1+\delta}} \abs{f_n^{(j)}(s)} ds+\max_{1\le j\le m}\int_{k^{-1+\delta}}^{1} \abs{f_n^{(j)}(s)} ds     \\
    :&=I_{1,1}+I_{1,2}.
 \end{aligned}
 $$
 We first handle $I_{1,2}$. Note that, for $\eta>\frac{2\gamma}{1/2-\gamma}$, we can always find a $v>\frac{1}{2+\eta}$ such that $v+\gamma<1/2$.
 Then, by Theorem \ref{Theorem: TQP DC}, as $N\to\infty$,
  $$
    I_{1,2} =o_P(1) (km)^{-1/2} \int_{k^{-1+\delta}}^{1} s^{-v-1/2-\gamma} ds =(km)^{-1/2}o_P(1).   
$$ 
 The term $I_{1,1}$ can be handled by Proposition \ref{Theorem: Corner} as follows. Choose $g(s)=1$. 
  Since $\gamma<1/2$ and $\eta >\max\set{0, \frac{2\gamma }{1/2-\gamma}}$, the conditions in Proposition \ref{Theorem: Corner} hold. The
  proposition yields that $I_{1,1}=(km)^{-1/2}o_P(1)$.
Hence, we obtain $ I_1=(km)^{-1/2}o_P(1)$ as $N\to \infty$. 
The term $I_2$ can be handled in a similar way  with choosing $g(s)=s$.
\end{proof}

\subsection{Distributed inference for the MLE}

 The MLE  for the extreme value index  and the scale parameter  based on the sample on machine $j$ $(\gamma_{mle}^{(j)}, \sigma_{mle}^{(j)})$, is defined as the solution of the following equations:
\begin{equation}\label{mle}
  \begin{aligned}
    \frac{1}{k}\sum_{i=1}^k &\frac{1}{\gamma^2} \log \suit{1+\frac{\gamma}{\sigma}\suit{X_{n-i+1,n}^{(j)} - X_{n-k,n}^{(j)}}} \\
    &-\suit{\frac{1}{\gamma}+1} \frac{(1/\sigma)\suit{X_{n-i+1,n}^{(j)}-X_{n-k,n}^{(j)}}}{1+(\gamma/\sigma)\suit{X_{n-i+1,n}^{(j)}-X_{n-k,n}^{(j)}}} = 0, \\
    & \sum_{i=1}^k \suit{\frac{1}{\gamma}+1} \frac{(\gamma/\sigma)\suit{X_{n-i+1,n}^{(j)}-X_{n-k,n}^{(j)}}}{1+(\gamma/\sigma)\suit{X_{n-i+1,n}^{(j)}-X_{n-k,n}^{(j)}}} = k.
  \end{aligned}
\end{equation}
The distributed MLE for the extreme value index and the scale parameter are defined as 
$$
\begin{aligned}
  \widehat{\gamma}_{mle}^{D}  =\frac{1}{m}\sum_{j=1}^m\widehat{\gamma}_{mle}^{(j)}, \quad 
  \widehat{\sigma}_{mle}^{D}  =\frac{1}{m}\sum_{j=1}^m\widehat{\sigma}_{mle}^{(j)}.
\end{aligned}
$$
The oracle MLE for the extreme value index and the scale parameter  $ (\widehat{\gamma}_{mle}^{Oracle}, \widehat{\sigma}_{mle}^{Oracle})$ are defined in a similar way by using the oracle sample.
\begin{corollary}\label{corollary:mle}
  Suppose that the distribution function F satisfies the second order
condition \eqref{Def:SOC:Uni} with $\gamma>-1/2$ and $\rho<0$. Assume that conditions \ref{Condition:u:k1}-\ref{Condition:u:k3} hold with $\eta>\max\suit{0,2\gamma,\frac{-2\gamma}{1+2\gamma}}$  . 
 Then, the distributed MLE for the extreme value index and the scale parameter achieve the oracle property, i.e., as $N\to\infty$,
$$
\begin{aligned}
  \sqrt{km}\suit{\widehat{\gamma}_{mle}^{D}  - \widehat{\gamma}_{mle}^{Oracle}}&=o_P(1).\\
  \sqrt{km} \ \frac{\widehat{\sigma}_{mle}^{D}  - \widehat{\sigma}_{mle}^{Oracle}}{a(n/k)}&=o_P(1).
\end{aligned}
$$
\end{corollary}

The proof is deferred to the Supplementary Material.
 Solving the likelihood equations  \eqref{mle} involves an optimization algorithm. The computation cost can be high when implementing an optimization algorithm for the oracle sample. We provide a simulation study to compare the computation cost of the oracle MLE and the distributed MLE in the Supplementary Material.

\subsection{Distributed inference for the  high quantile, endpoint and tail probability}
In this subsection, we show how to establish the oracle property of the estimators for the high quantile, endpoint and tail probability. In order to estimate these quantities, we need to estimate the extreme value index $\gamma$, the scale parameter $a(n/k)$ and the location parameter $b(n/k)$, see e.g. \citet[Chapter 4]{de2006extreme}. 
We focus on the PWM estimators for $\gamma$ and $a(n/k)$ as an example. Other estimators based on the POT approach can be treated in a similar way.

Based on the oracle sample, since $N/(km) = n/k$, one can estimate  $a(n/k)$ and  $b(n/k)$ as 
$$
\widehat{a}^{Oracle}\suit{\frac{n}{k}} = \frac{2P_NQ_N}{P_N-2Q_N}, \quad \widehat{b}^{Oracle}\suit{\frac{n}{k}} = X_{N-[km],N},
$$
see e.g. \cite{hosking1987parameter}.

We apply the DC algorithm to estimate $a(n/k)$ and  $b(n/k)$ based on distributed subsamples.
Define the distributed scale estimator as  
$$
\widehat{a}^{D}\suit{\frac{n}{k}}:=\frac{1}{m}\sum_{j=1}^m \widehat{a}^{(j)}\suit{\frac{n}{k}}=\frac{1}{m}\sum_{j=1}^m \frac{2P_n^{(j)}Q_n^{(j)}}{P_n^{(j)}-2Q_n^{(j)}},
$$
and the distributed location estimator as 
$$
\widehat{b}^{D}\suit{\frac{n}{k}} = \frac{1}{m}\sum_{j=1}^m X_{n-k,n}^{(j)}. 
$$
Following similar steps as in proving the oracle property of $\widehat{\gamma}_{PWM}^{D}$, we can show that, as $N\to\infty$,
$$
\sqrt{km}\ \frac{\widehat{a}^{D}\suit{\frac{n}{k}} -\widehat{a}^{Oracle}\suit{\frac{n}{k}}}{a\suit{\frac{n}{k}}}=o_P(1),
\quad 
\sqrt{km}\ \frac{\widehat{b}^{D}\suit{\frac{n}{k}}-\widehat{b}^{Oracle}\suit{\frac{n}{k}}}{a\suit{\frac{n}{k}}} = o_P(1).
$$
 
\subsubsection{High quantile}

Let $x\suit{p_N}:= U(1/p_N)$, where $p_N = o(k/n)$ as $N\to\infty$, be the quantile we want to estimate. In finance management, the high quantile is often referred to as value at risk, which is  the most prominent risk measure.
The detailed procedures of the distributed estimator for high quantile $x\suit{p_N}$ are given as follows:
\begin{itemize}
    \item On each machine $j$, we calculate $\widehat{\gamma}_{PWM}^{(j)},\widehat{a}^{(j)}\suit{\frac{n}{k}}, X_{n-k,n}^{(j)}$ and transmit these values to the central machine.
    \item On the central machine, we take the average of the $\widehat{\gamma}_{PWM}^{(j)},\widehat{a}^{(j)}\suit{\frac{n}{k}}, X_{n-k,n}^{(j)}$ statistics collected from the $m$ machines to obtain $\widehat{\gamma}_{PWM}^{D},\widehat{a}^{D}\suit{\frac{n}{k}}, \widehat{b}^{D}\suit{\frac{n}{k}} $.
    \item On the central machine,  we estimate $x\suit{p_N}$ with $p_N\to 0$ by 
    \begin{equation}\label{Eq: DC high Quantile}
      \widehat{x}^{D}(p_N) = \widehat{b}^{D}\suit{\frac{n}{k}}+\widehat{a}^{D}\suit{\frac{n}{k}} \frac{\suit{\frac{k}{np_N}}^{\widehat{\gamma}_{PWM}^{D}}-1}{\widehat{\gamma}_{PWM}^{D}}.
    \end{equation}

\end{itemize}
The oracle high quantile estimator $\widehat{x}^{Oracle}(p_N)$ is defined in an analogous way as  $\widehat{x}^{D}(p_N)$, with replacing $\widehat{\gamma}_{PWM}^{D}$, $\widehat{a}^{D}\suit{\frac{n}{k}}$ and $\widehat{b}^{D}\suit{\frac{n}{k}}$  by  $\widehat{\gamma}_{PWM}^{Oracle}$, $\widehat{a}^{Oracle}\suit{\frac{n}{k}}$ and $\widehat{b}^{Oracle}\suit{\frac{n}{k}}$ in \eqref{Eq: DC high Quantile}, respectively.
Following the lines of the proof for the asymptotics  of  the oracle high quantile estimator, we can obtain the asymptotic normality of $\widehat{x}^{D}(p_N)$. Moreover, since $\widehat{\gamma}_{PWM}^{D}$, $\widehat{a}^{D}\suit{\frac{n}{k}}$ and $\widehat{b}^{D}\suit{\frac{n}{k}}$ possess the oracle property, $\widehat{x}^{D}(p_N)$ also  achieves the oracle property due to applying the  Cram\'er delta method. We present the result in the following corollary while omitting the proof.
\begin{corollary}
    Assume the same conditions  as in Corollary \ref{Corollary: PWM}. Suppose that $np_N = o(k)$ and $\log (Np_N) = o(\sqrt{km})$ as $N\to\infty$. 
    Then,  as $N\to\infty$,
    $$
\sqrt{km}\frac{\widehat{x}^{D}(p_N)-\widehat{x}^{Oracle}(p_N)}{a\suit{\frac{n}{k}}q_{\gamma}(d_N)}=o_P(1),
    $$
    where $d_N = k/(np_N)$ and for $t>1$,
    $
    q_{\gamma}(t):=\int_{1}^t s^{\gamma-1}\log s ds.
    $

\end{corollary}

\subsubsection{Endpoint}

Next, we consider the problem of estimating the endpoint of the distribution function $F$. Assume that $F\in D(G_{\gamma})$ for some  $\gamma<0$. In this case the endpoint $x^* = \sup\set{x: F(x)<1}$ is finite. The endpoint can be treated as a specific case of quantile by regarding  $p_N$ as  $0$. The distributed endpoint estimator can be defined as 
$$
\widehat{x}^{*,D} =\widehat{b}^{D}\suit{\frac{n}{k}} - \frac{\widehat{a}^{D}\suit{\frac{n}{k}}}{\widehat{\gamma}_{PWM}^{D}}.
$$
The definition of the oracle endpoint estimator $\widehat{x}^{*,Oracle}$ is in an analogous way.
Again, the distributed endpoint estimator achieves the oracle property  as in the following corollary.
\begin{corollary}\label{Corrollary: endpoint}
  Assume the same conditions  as in Corollary \ref{Corollary: PWM} and $\gamma<0$.   Then, as $N\to\infty$,
  $$
\sqrt{km}\frac{\widehat{x}^{*,D}-\widehat{x}^{*,Oracle}}{a\suit{\frac{n}{k}}} =o_P(1).
  $$
\end{corollary}

\subsubsection{Tail probability}
Lastly, we consider the dual problem of estimating the high quantile: given a large value of $x_N$, how to  estimate $p(x_N) = 1-F(x_N)$ under the distributed inference setup. The detailed procedures for estimating the tail probability are similar to that for  estimating the high quantile, except that on the central machine, we 
estimate  the tail probability $p(x_N)$ by 
$$
\widehat{p}^{D}(x_N) = \frac{k}{n}\set{\max\suit{0,1+\widehat{\gamma}_{PWM}^{D}\frac{x_N-\widehat{b}^{D}\suit{\frac{n}{k}}}{\widehat{a}^{D}\suit{\frac{n}{k}}}}}^{-1/\widehat{\gamma}_{PWM}^{D}}.
$$
The definition of the oracle tail probability estimator $\widehat{p}^{Oracle}(x_N)$ is  in an analogous way.
 Note that $\widehat{p}^{Oracle}(x_N)$  is valid only for $\gamma>-1/2$ (cf. Remark 4.4.3 in \cite{de2006extreme}). The oracle property of $\widehat{p}^{D}(x_N)$ is established in the following corollary.
\begin{corollary}\label{Corrolary: tail probability}
  Assume the same conditions as  in Corollary \ref{Corollary: PWM} and  $\gamma \in (-1/2, 1/2)$.  Denote  $d_N = \frac{k}{np(x_N)}$ and $ w_{\gamma}(t) = t^{-\gamma}\int_{1}^{t} s^{\gamma-1}\log s ds$ for $t>0$. 
  Suppose that $d_N \to \infty$ and  $w_{\gamma}(d_N) = o(\sqrt{km})$ as $N\to\infty$, then 
  $$
  \frac{\sqrt{km}}{w_{\gamma}(d_N)} \frac{\widehat{p}^{D}(x_N)-\widehat{p}^{Oracle}(x_N)}{p(x_N)} =o_P(1).
  $$
\end{corollary}

\section{Heterogeneous subsample sizes}\label{sec:heter}

In this section, we  extend our results to the case of heterogeneous subsample sizes.  We assume that  the $N$ observations are  distributedly stored in  $m$ machines  with $n_j$ observations in machine $j$, $j=1,\dots,m$, i.e. $N=\sum_{j=1}^m n_j$. Moreover,  we assume that all $n_j$, $j=1,\dots,m$ diverge in the same order. Mathematically,  that is, 
\begin{equation}\label{Eq heter sizes}
  c_1\le \min_{1\le j \le m} n_j m/N \le \max_{1\le j \le m} n_j m/N \le c_2  
\end{equation}
 for some positive constants $c_1$ and $c_2$ and all $N\ge 1$.

The tail empirical process based on the observations in machine $j$ is now defined as 
$$
Y_{n_j,k_j}^{(j)}(x) = \frac{n_j}{k_j}\bar{F}_{n_j}^{(j)}\set{a_0\suit{\frac{n_j}{k_j}}x+b_0\suit{\frac{n_j}{k_j}}}, \quad j=1,\dots,m,
$$
where $\bar{F}_{n_j}^{(j)}:=1-F_{n_j}^{(j)}$ and $F_{n_j}^{(j)}$ denotes the  empirical distribution function
based on the observations in machine $j$.  


We choose $k_j, j=1,\dots,m$ such that the ratios  $k_j/n_j$ are homogenous across all the $m$ machines, i.e., 
\begin{equation}\label{Eq Hetergenous k}
k_1/n_1  = \cdots = k_m/n_m.
\end{equation}
Denote $K = \sum_{j=1}^m k_j$, clearly $k_j/n_j = K/N, j=1,\dots,m$. Then, we have that, 
$$
Y_{N,K}(x) = \sum_{j=1}^m \frac{n_j}{N} Y_{n_j,k_j}^{(j)}(x).
$$
In other words, the oracle tail empirical process is a weighted average of the  tail empirical processes based on the distributed
subsamples, where the weights equal to the fraction of the observations on each machine.

Following similar steps as in proving Theorem \ref{Theorem: EP Divide and Conquer}, we obtain the following result.
\begin{theorem}
  Assume the same conditions as in Theorem \ref{Theorem: EP Divide and Conquer} and conditions \eqref{Eq heter sizes} and  \eqref{Eq Hetergenous k}.
  Then under proper Skorokhod construction,  there exist $m$ independent sequences of Brownian motions $\set{W_{n_j}^{(j)}}, j=1,\dots,m$, such that for any $v\in ((2+\eta)^{-1},1/2)$, as $N\to\infty$,
  $$
  \begin{aligned}
    &\max_{1\le j\le m} \sup_{x\in \mathbb{D}} \set{z(x)}^{v-1/2} \Big|\sqrt{k_jm}\set{Y_{n_j,k_j}^{(j)}(x)-z(x)} \\
    &-\sqrt{m}W_{n_j}^{(j)}\set{z(x)}-\sqrt{k_jm}A_0(N/K)\set{z(x)}^{1+\gamma}\Psi\set{1/z(x)}\Big|  = o_P(1).
  \end{aligned}
  $$
  Moreover,  as $N \to \infty$,
  $$
  \begin{aligned}
    \sup_{x\in \mathbb{D}} &\set{z(x)}^{v-1/2}\Big|\sqrt{K}\set{Y_{N,K}(x)-z(x)} \\
    &-W_N\set{z(x)}-\sqrt{K}A_0(N/K)\set{z(x)}^{1+\gamma}\Psi\set{1/z(x)}\Big| =o_P(1),
  \end{aligned} 
  $$
  where  $W_N = \sum_{j=1}^m \sqrt{\frac{n_j}{N}} W_{n_j}^{(j)}$ is also a  Brownian motion.
\end{theorem}

Similar results hold for the tail quantile processes  as in Theorem \ref{Theorem: TQP DC}. Eventually, we can re-establish the oracle property of the distributed estimators as follows. 
Suppose the oracle estimator is based on $K$ top order statistics in the oracle sample. On each machine, we use the top $k_j  = (n_j/N) K$ order statistics in the estimation. 
 By taking  a weighted average of the estimates from all machines using the weights $n_j/N, j=1,\dots,m$, to obtain the distributed estimator,  the oracle property holds under the same conditions as in  the homogenous case  with similar proofs.

 \section{Real Data Application}\label{sec:realdata}

 We use a  dataset containing car insurance claims in five states of United States:  Iowa ($n_1 = 2601$), Kansas ($n_2 = 798$), Missouri ($n_3 = 3150$), Nebraska ($n_4 = 1703$), and Oklahoma ($n_5 = 882$). 
The total sample size is $N = 9134$. We work under a hypothesis scenario: each state cannot share its own data to others, but they are willing to share their statistical results. Then one can apply a DC algorithm for conducting extreme value statistics.
 Our target is to estimate the common extreme value index of the total claim amount. We consider the MLE instead of  the Hill estimator considered by  \cite{chen2021distributed} since we do not assume heavy tail at the first place.

Let $K = \sum_{j=1}^5 k_j$ be the total number of exceedances used by the  five states. 
As suggested in Section \ref{sec:heter}, we choose $k_j$ as 
$
k_j = [K \frac{n_j}{N}],
$
and apply the MLE for each of the five states to obtain $\widehat{\gamma}_{mle}^{(j)}$, $ j=1,2,\dots, 5$. Then, we combine these five estimates to obtain the distributed  MLE by 
$$
\widehat{\gamma}_{mle}^{D} = \sum_{j=1}^5 \frac{n_j}{N} \widehat{\gamma}_{mle}^{(j)}.
$$

\begin{figure}
  \centering
  \includegraphics[width = 0.8\textwidth]{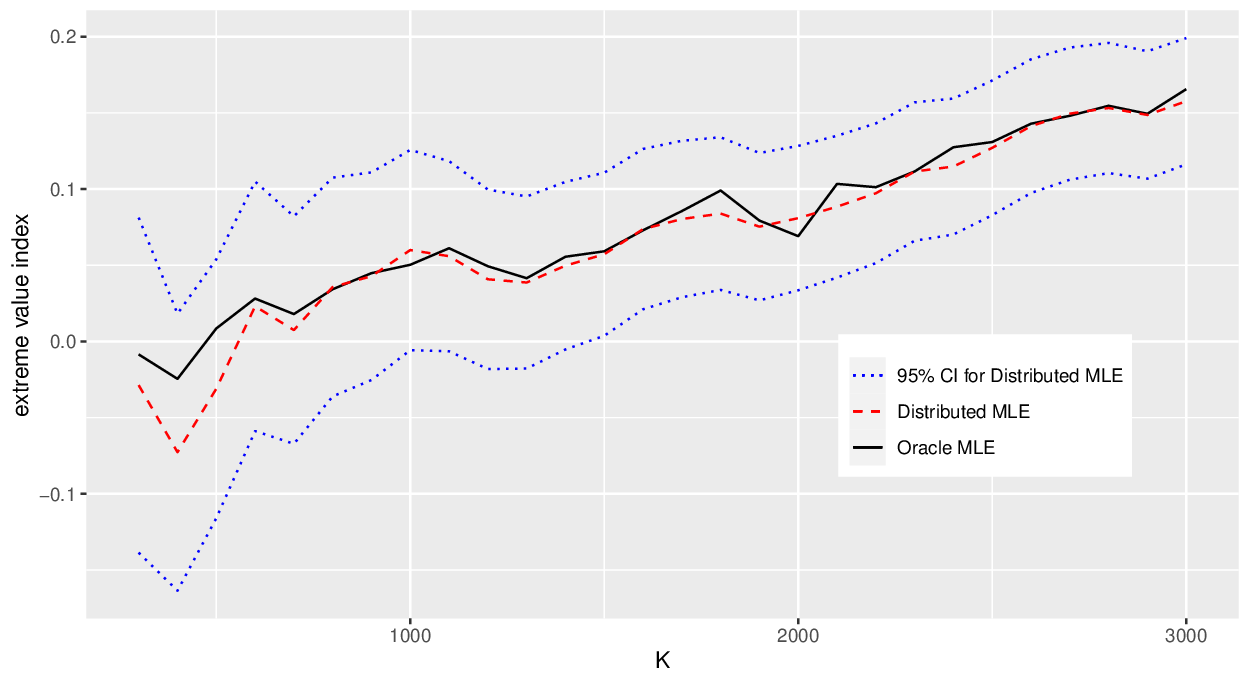}
  \caption{Car insurance data.}
  \label{Fig Real data}
\end{figure}

The distributed MLE is plotted against different values of $K$ in Figure \ref{Fig Real data}, along with its 95\% confidence interval.  We also plot the oracle MLE  in this figure.   The distributed MLE is close to the oracle MLE for almost all levels of $K$ and the oracle  MLE  always falls into the 95\% confidence interval constructed based on  the distributed MLE. 

By choosing $K=1000$, we obtain that the distributed MLE for the extreme value index is about 0.05. And we cannot reject the hypothesis that the extreme value index is $0$ under the 5\% significance level for this choice of $K$. This result shows that the insurance claims may not be heavy tailed. In turn, the distributed Hill estimator adopted in \cite{chen2021distributed} may not be  suitable for this application.

\section{Discussion}\label{sec:discuss}

In this paper, we investigate the problem of distributed inference in extreme value analysis when the oracle sample $\set{X_1,X_2,\dots,X_N}$ are i.i.d.. In fact, the assumption that all the data are drawn from the same distribution can be relaxed. In real applications, observations from different machines may follow different distributions, but nevertheless share some common properties such as the extreme value index. 

We assume that  all observations are independent, but only observations on the same machine follow the same distribution. 
Denote the common distribution function of the observations in machine $j$ as $F_{n,j}, j=1,\dots,m$. We assume that, there exists a continuous function $F$ which satisfies the second order condition \eqref{Def:SOC:Uni} with $\gamma>0$. 
In addition, assume that the series of constants $\set{c_{n,j}}_{1\le j \le m}$ satisfies that $0< \underline{c} \le c_{n,j}\le \bar{c}<\infty$ for all $1\le j \le m$ and $n\in \mathbb{N}$, and $A_1(t)$ is  
 a positive regularly varying function  with index $\tilde{\rho}< 0$   such that as $t\to\infty$,
$$
\sup_{m\in \mathbb{N}}\max_{1\le j\le m}\abs{\frac{1-F_{n,j}(t)}{1-F(t)} -c_{n,j}}  = O(A_1(t)).
$$
By restricting  that $\sqrt{km} A_1(n/k) \to 0$, 
\cite{chen2021distributed} gives a theoretical proof for the asymptotic theories of the distributed Hill estimator.  Following similar steps,  we can also handle  tail empirical processes and tail quantile processes. The details are omitted.


\bigskip
\begin{center}
{\large\bf Supplementary Material}
\end{center}

\begin{description}

\item[Appendix.pdf:]  The Supplementary Material contains all the technical proofs and simulation studies.

\end{description}

\bibliographystyle{apalike} 
\bibliography{mybib} 

\end{document}